\documentclass[11pt,a4paper]{article}
\usepackage[utf8]{inputenc}
\usepackage{amssymb}
\usepackage{amsmath}
\usepackage{amsthm}
\usepackage[margin=1in]{geometry}
\usepackage{framed}
\usepackage{hyperref}
\hypersetup{colorlinks=true,citecolor=blue,linkcolor=blue,urlcolor=blue}
\usepackage{xcolor}

\usepackage[T1]{fontenc}

\newtheorem{theorem}{Theorem}
\newtheorem{lemma}[theorem]{Lemma}

\newtheorem{observation}[theorem]{Observation}
\theoremstyle{definition}
\newtheorem{definition}[theorem]{Definition}
\newtheorem{example}[theorem]{Example}

\newcommand{\PCSP}{\operatorname{PCSP}}

\newcommand{\Pol}{\operatorname{Pol}}

\newcommand{\id}{\operatorname{id}}

\newcommand{\Aff}{\operatorname{Aff}}

\newcommand{\ar}{\operatorname{ar}}

\newcommand{\LP}{\operatorname{LP}}

\newcommand{\bA}{\mathbf{A}}
\newcommand{\bB}{\mathbf{B}}
\newcommand{\bC}{\mathbf{C}}
\newcommand{\bD}{\mathbf{D}}
\newcommand{\bX}{\mathbf{X}}
\newcommand{\bF}{\mathbf{F}}

\newcommand{\ZZ}{\mathbb{Z}}
\newcommand{\NN}{\mathbb{N}}
\newcommand{\QQ}{\mathbb{Q}}

\newcommand{\FF}{\mathbb{F}}
\newcommand{\Mm}{\mathcal{M}}
\newcommand{\Nn}{\mathcal{N}}
\newcommand{\Qconv}{\mathcal{Q}_{conv}}
\newcommand{\MBA}{\Mm_{\mathrm{BLP+Aff}}}

\title{The Power of the Combined Basic LP and Affine \\ Relaxation for Promise
CSPs\thanks{An extended abstract of part of this work (by the first two authors) appeared in the
\emph{Proceedings of the 31st Annual ACM-SIAM Symposium on Discrete Algorithms
(SODA'20)}~\cite{bg20}.}}
\author{Joshua Brakensiek\thanks{Stanford University, Stanford, CA 94305, USA.
Email: {\tt jbrakens@stanford.edu}. Research supported in part by an REU
supplement to NSF CCF-1526092 and a NSF Graduate Research Fellowship.} \and
Venkatesan Guruswami\thanks{Computer Science Department, Carnegie Mellon
University, Pittsburgh, PA 15213. Email: {\tt venkatg@cs.cmu.edu}. Research
supported in part by NSF grants CCF-1814603 and CCF-1908125.} \and Marcin
Wrochna\thanks{University of Oxford, UK. Email: {\tt
marcin.wrochna@cs.ox.ac.uk}. Research supported by funding from the European
Research Council (ERC) under the European Union's Horizon 2020 research and
innovation programme (grant agreement No 714532).}
\and
Stanislav \v{Z}ivn\'y\thanks{University of Oxford, UK. Email: {\tt
standa.zivny@cs.ox.ac.uk}. Research supported by a Royal Society
University Research Fellowship and by funding from the European
Research Council (ERC) under the European Union's Horizon 2020 research and
innovation programme (grant agreement No 714532).}
}
\date{}

\begin{document}

\date{}

\maketitle

\begin{abstract}
  In the field of constraint satisfaction problems (CSP), promise CSPs
  are an exciting new direction of study. In a promise CSP, each
  constraint comes in two forms: “strict” and “weak,” and in the
  associated decision problem one must distinguish between being able
  to satisfy all the strict constraints versus not being able to
  satisfy all the weak constraints. The most commonly cited example of
  a promise CSP is the approximate graph coloring problem—which has
  recently seen exciting progress
  \cite{BulinKrokhinOprsal2019,wz20}
  benefiting from a systematic algebraic
 approach to promise CSPs based on ``polymorphisms,'' operations
  that map tuples in the strict form of each constraint to tuples in
  the corresponding weak form.

\smallskip
  In this work, we present a simple algorithm which in polynomial time
  solves the decision problem for all promise CSPs that admit
  infinitely many \emph{symmetric} polymorphisms, which are invariant under arbitrary coordinate permutations. This generalizes previous
  work of the first two authors \cite{BrakensiekGuruswami2019}.  We also extend
  this algorithm to a more general class of \emph{block-symmetric}
  polymorphisms. As a corollary, this single algorithm solves all
  polynomial-time tractable Boolean CSPs simultaneously. These results
  give a new perspective on Schaefer’s classic dichotomy theorem and shed
  further light on how symmetries of polymorphisms enable algorithms.
  Finally, we show that block symmetric polymorphisms are not only sufficient
  but also necessary for this algorithm to work, thus establishing its precise
  power.
\end{abstract}

\newcommand{\PCSPD}{\operatorname{PCSP-Decision}}
\newcommand{\PCSPS}{\operatorname{PCSP-Search}}

\section{Introduction}

A central challenge in the theory of algorithms is to understand the
mathematical structure (or lack thereof) that governs the efficient
tractability (or intractability) of a computational problem. For the
class of constraint satisfaction problems (CSP), a rich algebraic
theory culminating in the recent resolution of the Feder-Vardi
dichotomy conjecture~\cite{DBLP:journals/siamcomp/FederV98} in \cite{DBLP:conf/focs/Bulatov17,DBLP:conf/focs/Zhuk17} has
established a striking link between problem structure and its
tractability. In particular, a CSP is efficiently solvable if and only if its
defining relations admit an ``interesting'' polymorphism. Informally,
a polymorphism is a function whose component-wise action preserves
membership in the relations defining the CSP, and ``interesting''
means that the function obeys some \emph{non-trivial identities}. As an
example, for the (efficiently solvable) CSP corresponding to linear
equations over a ring $R$, the $3$-ary function $f(x,y,z) = x-y+z$ is
a polymorphism (capturing the fact that if $v_1,v_2,v_3$ are solutions
to a linear system, then so is $v_1-v_2+v_3$), and it obeys the
so-called Mal'tsev identity $f(x,y,y) = f(y,y,x) = x$ for all $x,y \in
R$. Indeed, generalizing Gaussian elimination, any CSP with such a
Mal'tsev polymorphism is efficiently tractable~\cite{DBLP:journals/eccc/ECCC-TR02-034,BulatovDalmau2006}.

Recently, an exciting new direction of study has emerged in the rich
backdrop of the complexity dichotomy for CSPs. This concerns a vast
generalization of the CSP framework to the class of \emph{promise
  constraint satisfaction problems} (PCSP). In a promise CSP, each
constraint comes in two forms: ``strict'' and ``weak.'' Given an
instance of a PCSP, one must distinguish between being able to satisfy
all the strict constraints versus not being able to satisfy all the
weak constraints. (This is the \emph{decision} version; in the
\emph{search} version, given an instance with a promised assignment
satisfying the strong form of the constraints, one seeks an assignment
satisfying the weak form of the constraints.) A prime example of a
PCSP is the approximate graph coloring problem, where one seeks to
color a graph using more colors than its chromatic number.

The formal study of promise CSPs originated in
\cite{DBLP:journals/siamcomp/AustrinGH17} who classified the
complexity of a PCSP called $(2+\epsilon)$-SAT. They further defined
an extension of polymorphisms to the promise setting and postulated
that the structure of those polymorphisms might govern the complexity
of a PCSP. (This extension of polymorphisms to the promise setting is
quite natural, requiring that the operation map tuples obeying the
strict form of a constraint to a tuple satisfying its weak form.)
Building on the impetus of \cite{DBLP:journals/siamcomp/AustrinGH17},
Brakensiek and Guruswami systematically studied PCSPs under
the polymorphic lens and established promising links to the
universal-algebraic framework developed for
CSPs~\cite{DBLP:conf/soda/BrakensiekG18,BrakensiekGuruswami2019}. It
emerged from these works that a rich enough family of polymorphisms
leads to efficient algorithms, whereas severely limited polymorphisms
are a prescription for hardness. However, unlike for CSPs, there is no
sharp transition between these cases --- the significant difficulty
being that, unlike for CSPs, polymorphisms for PCSPs are \emph{not}
closed under composition and lack the rich algebraic structure of a
\emph{clone} (c.f., \cite{BartoKrokhinWillard2017}). This nascent
algebraic theory for PCSPs was lifted to a more abstract level in
\cite{BulinKrokhinOprsal2019,BartoBulinKrokhinEtAl2019} and also led
to concrete breakthroughs in approximate graph
coloring/homomorphisms~\cite{BulinKrokhinOprsal2019,KrokhinOprsal2019,wz20}.
In
particular, while previous
works~\cite{DBLP:conf/soda/BrakensiekG18,BrakensiekGuruswami2019}
focused on the actual \emph{form} of the polymorphisms, the results of
\cite{BulinKrokhinOprsal2019} reveal that it is not the polymorphisms
themselves, but rather solely the \emph{identities} they
satisfy, that capture the complexity of the associated PCSP, extending a
similar phenomenon known earlier for
CSPs~\cite{BartoOprsalPinsker2018}.

This work concerns the theme of designing algorithms for PCSP based on
a rich enough family of polymorphisms. Our main result is that the
decision version of an arbitrary PCSP admitting an \emph{infinite
  family of symmetric polymorphisms} --- i.e., polymorphisms which are
invariant under any permutation of inputs --- is tractable (see
Theorem~\ref{thm:sym}). Our result also extends to the case of
\emph{block-symmetric} polymorphisms (see
Theorem~\ref{thm:block-sym}). That is, the coordinates can be
partitioned into ``blocks'' such that the function is invariant under
permutations within each block. Notably, in the block-symmetric case
the algorithm is identical--only the analysis changes. Furthermore,
the number of blocks is irrelevant, the only assumption we need is
that the minimum block size can be made arbitrarily large.
Our final result (Theorem~\ref{thm:char}) shows that block-symmetry is not
only sufficient but also necessary for our algorithm to work. In fact,
Theorem~\ref{thm:char} also establishes that without
loss of generality one can assume that there are only two blocks of symmetric
coordinates.

Further our algorithm is very simple --- it checks if the canonical
linear programming (LP) relaxation of the PCSP is feasible, and if so,
it further checks if a slight adaptation of a canonical \emph{affine
  relaxation} is feasible. The algorithm outputs satisfiable if both
these relaxations are feasible. The polymorphisms are not used in the
algorithm itself and only enter the analysis. The analysis is short
but subtle --- if we had symmetric polymorphisms of all arities then
it is known that the basic LP relaxation itself correctly decides
satisfiability, as one can round the fractional solution to a
satisfying assignment using the polymorphism after clearing
denominators of the fractional
solution~\cite{DBLP:conf/innovations/KunOTYZ12,BartoKrokhinWillard2017}. If
polymorphisms only exist of certain arities (e.g., all odd majorities),
then the LP alone doesn't suffice (e.g., \cite{DBLP:conf/innovations/KunOTYZ12}). We solve a linear system over the
integers corresponding to the affine relaxation which lets us adjust
the LP solution to match the arity at which a polymorphism exists. As
a subtle twist, the affine relaxation is not of the original PCSP, but
rather a \emph{refinement} of the CSP which results from throwing out
assignments to constraints which were ruled out by the basic LP.

It should be pointed out that we only solve the \emph{decision}
version of the PCSP, and \emph{not} the search version. Unlike CSPs,
for promise CSP there is no known reduction from search to decision,
even for special cases like approximate graph coloring. Our work might
be indicative of the subtle relationship between the search and
decision problems for promise CSPs.

We now compare our result here with the previous
work~\cite{BrakensiekGuruswami2019} where an algorithm was given to
solve (the search version of) any PCSP that admits an infinite family
of \emph{structured symmetric polymorphisms}. Examples of such
structured families include threshold and threshold-periodic
polymorphisms. The value of a threshold polymorphism (for a Boolean PCSP)
depends on whether the fraction of $1$s in the input belongs in a finite number of intervals.  (A basic example consists of Majority
functions of odd arities, which are polymorphisms for 2-SAT.)  A
threshold-periodic polymorphism can have a periodic behavior depending
on which interval the Hamming weight belongs to --- for example it can
be Majority for relative weights in $(1/3,2/3)$ and parity outside
this interval. More generally, one can generalize to the non-Boolean
case, as well as for the block-symmetric case, via regional
polymorphisms whose value depends on the geometric region in which the
vector of frequencies of the inputs to the polymorphisms lies. Due to
this geometric interpretation, \cite{BrakensiekGuruswami2019} assumes
a fixed number of blocks (corresponding to a fixed dimension), whereas
our new algorithm and analysis is independent of the number of
blocks. The algorithm was a combination of solving the LP relaxation
(albeit over a special ring like $\mathbb{Z}[\sqrt{2}]$ rather than
the rationals) and the affine relaxation over a large enough finite
ring.  The analysis relied on the special \emph{structure} of the
polymorphisms (beyond their full symmetry). In contrast, our result
here is more general, and only requires the polymorphism to be a
symmetric function --- its exact specifics or structure do not
matter. It is encouraging that our methodology is consistent with the 
algebraic result in \cite{BulinKrokhinOprsal2019} that the symmetries
possessed by the polymorphisms capture the complexity of the PCSP.

Our result and methods have significance even for normal (non-promise)
CSPs. For instance, we get a single unified algorithm to solve all
non-trivial tractable cases of Boolean CSPs in Schaefer's classic
dichotomy theorem~\cite{Schaefer:1978}, namely 2-SAT, Horn-SAT (or its
dual), and Mod-2 Linear Equations. The two main techniques to solve
CSPs are local propagation based algorithms (which work for the
so-called \emph{bounded-width}
CSPs~\cite{BartoKozik:2014,DBLP:conf/innovations/KunOTYZ12}, etc.) and
Gaussian elimination (which is a global algorithm that works for
linear equations). The major difficulty in proving the full CSP
dichotomy was tackling the complicated ways in which these two very
different algorithms might need to be interlaced to solve a general
CSP. It is our hope that this work serves as an impetus toward the
potential discovery of a more modular CSP algorithm that incorporates
together linear programming or its extensions (like Sherali-Adams, or
semidefinite programming) and linear equation solving. In this light,
it is encouraging that full symmetry of the polymorphisms, which is
indeed a strong assumption, is not the limit of our techniques, which
also extend to the block-symmetric case.

To put this work in further context, except for
\cite{BrakensiekGuruswami2019} as mentioned previously, nearly all
works in the PCSP
literature~\cite{DBLP:journals/siamcomp/AustrinGH17,DBLP:conf/soda/BrakensiekG18,FicakKozikOlsakEtAl2019}
focus primarily on the structure of the \emph{relations}. In
particular,
\cite{DBLP:conf/soda/BrakensiekG18,FicakKozikOlsakEtAl2019}
characterized the complexity of all \emph{Boolean symmetric relations}
(rather than \emph{symmetric polymorphisms}) which encompass many of
the known tractable cases of Boolean PCSP. As classified by
\cite{FicakKozikOlsakEtAl2019}, all the relevant tractable
polymorphisms are either symmetric functions or one special case of
block-symmetric known as alternative threshold (and variants). Thus,
in the context of PCSPs, the algorithm in this paper supersedes these
previous works. See Section~\ref{sec:block} for further discussion.

\section{Notation}

We let any finite set $A$ or $B$ denote a \emph{domain}.
A \emph{relation} is a subset $R \subseteq A^k$ for any positive integer $k$;
we denote $\ar(R) := k$. We define a \emph{signature} $\tau$  to be a set of symbols such that each $R \in \tau$ has a positive integer \emph{arity} $\ar(R)$. 

A \emph{relational structure} with signature $\tau$, denoted by $\bA := \{R^\bA \subseteq A^{\ar(R)} : R \in \tau\}$, is an indexed set of relations over $A$. A
\emph{homomorphism} between structures $\bA = \{R^\bA : R \in \tau\}$
and $\bB = \{R^\bB : R \in \tau\}$ with the same signature is a map $\sigma : A \to B$ such
that $\sigma(R^\bA) \subseteq R^\bB$ for all $R \in \tau$ (where $\sigma$ is
applied to a tuple component-wise).

Two relational structures for which there exists a homomorphism from the first to
the second is called a \emph{promise template} and is denoted as
$(\bA, \bB)$.

\subsection{PCSP: Decision and Search}

Consider a promise template $(\bA, \bB)$ with signature $\tau$.
An \emph{instance} $\bX$ of the \emph{promise constaint satisfaction problem} $\PCSP(\bA, \bB)$
consists of a set of variables $X := \{x_1, \hdots, x_n\}$,
and a set of constraints $c_1, \hdots, c_m$,
where $c_j := (R_j, \bar{x}^j)$, where $R_j$ is a symbol in $\tau$ and $\bar{x}^j$ is a tuple of arity $\ar(R_j)$.
We say that $\bX$ is \emph{satisfiable in $\bA$} if one can assign to every variable $x_i$ ($i \in [n]$) a value $\sigma(x_i)$ in the domain $A$
so that for every constraint $c_j = (R_j, \bar{x}^j)$ ($j \in [m]$),
the tuple $\sigma(\bar{x}^j)$ (with $\sigma$ applied component-wise) is in $R_j^\bA$.
Equivalently, $\bX$ can be described as a relational structure with domain $X$
and relations $R^\bX = \{ \bar{x} \in X^{\ar(R)} \colon \exists_{j \in [m]}\ c_j = (R,\bar{x}) \}$;
a satisfying assignment is then the same as a homomorphism from $\bX$ to $\bA$.
If $\bX$ is satisfiable in $\bA$, then it is satisfiable in $\bB$ (because the satisfying assignment can be composed with the homomorphism from $\bA$ to $\bB$).

We let $\PCSPD(\bA, \bB)$ denote the \emph{decision} problem of
distinguishing instances satisfiable in $\bA$ from those unsatisfiable in $\bB$
(with the promise that the input instance falls into one of these two disjoint cases).
We let $\PCSPS(\bA, \bB)$ denote the
\emph{search} problem of finding an explicit homomorphism from $\bX$ to $\bB$,
with the promise that a homomorphism from $\bX$ to $\bA$ exists.

\subsection{Polymorphisms}

A \emph{polymorphism} of $(\bA, \bB)$ of arity $L \in \NN$ is a map $f : A^L \to B$ such
that for all $R \in \tau$, $ R^\bB \supseteq f(R^\bA, \hdots, R^\bA)$ where we define the latter to be $\{(f(x^{(1)}_1, \hdots,
x^{(L)}_1), \hdots, f(x^{(1)}_{\ar(R)}, \hdots, x^{(L)}_{\ar(R)})) :
x^{(1)}, \hdots, x^{(L)} \in R^\bA\}.$
In other words, consider any $A^{L \times \ar(R)}$ matrix $M$, where
each \emph{row} is a satisfying assignment in $R^\bA$. Let
$y \in B^{\ar(R)}$ be the result of applying $f$ to each
\emph{column} of $M$. Then, $y \in R^B$.
We let $\Pol(\bA, \bB)$ denote the set of
polymorphisms of $(\bA, \bB)$ (of all arities).

A map $f : A^{L} \to B$ is said to be \emph{symmetric} if for all
$\pi \in S_{L}$ (the symmetric group on $L$ elements), $f(x_1, \hdots, x_{L}) = f(x_{\pi(1)}, \hdots,
x_{\pi(L)})$.

\subsection{Basic LP and Affine Relaxation} 
\label{subsec:relax}

As is well-studied in the CSP literature (e.g.,
\cite{RaghavendraSteurer2009,ThapperZivny2017}), we consider the
canonical linear programming relaxation of a CSP instance, often referred to as
the ``Basic LP'' or ``BLP.'' For our CSP instance $\bX$, we represent the
assignment $X \to A$ of a variable by a (rational) probability distribution of
weights $\{w_i(a)\}_{a \in A}$ summing to $1$. We also have a
probability distribution over the satisfying assignments to each
constraint, which we denote as $p_j({y})$, where $j \in [m]$ is the
index of the constraint and ${y}\in R_j^\bA$ is the potential assignment.
Finally, the marginal distribution of a variable $x_i$ in any constraint has to equal $w_i$.
Explicitly,
the linear constraints are as follows.
\allowdisplaybreaks
\begin{align} w_i(a) &\ge 0&&\text{for all $i\in[n]$ and $a\in A$} \label{eq:1}\\
p_j({y}) &\ge 0&&\text{for all $j\in[m]$ and ${y}\in R_j^\bA$} \label{eq:2}\\ \sum_{a \in
A}w_i(a) &= 1&&\text{for all }i\in[n]\label{eq:3}\\ \sum_{{y} \in R_j^\bA} p_j({y}) &=
1&&\text{for all }j \in [m]\label{eq:4}\\ \sum_{\substack{{y} \in R_j^\bA\\y|_i = a}}
p_j({y}) &= w_i(a)&&\parbox{4.8cm}{for all $i\in[n], a \in A, j \in [m]$
with $x_i$ in $\bar{x}^j$}.\label{eq:5}
\end{align}
Here $y|_i = a$ denotes that setting $\bar{x}^j={y}$ sets $x_i=a$
(that is, if $x_i$ is the $k$-th variable of the tuple $\bar{x}^j$, then $a = y_k$).
We let $\LP_\QQ(\bX, \bA)$ denote the rational polytope of
solutions. By a theorem of \cite{grotschel2012geometric} (c.f., \cite{BrakensiekGuruswami2019}), we can efficiently find a \emph{relative interior point} in this
polytope. In particular, at such a point, each coordinate is
positive if and only if it is positive at some point in the
polytope.\footnote{For our specialized LP, we do not need such a
  hammer. We can instead solve the LP repeatedly, each time maximizing a
  different variable as the objective function--a similar idea appears in~\cite{DBLP:conf/soda/BrakensiekG18}. Averaging the results
  would then yield a solution such that each variable is positive if
  and only if it is positive in some LP solution.}

In addition to the Basic LP, we also consider the affine relaxation of
a Promise CSP. In essence we solve the same linear system, but instead
of enforcing each variable to be a nonnegative rational, we enforce
that it is an integer (possibly negative). This can be solved in
polynomial time via \cite{kannan1979} (see also \cite{BrakensiekGuruswami2019} for a more
detailed discussion of this approach). We let $r_i(a) \in \mathbb Z$ replace $w_i(a)$ for all $a \in A$ and
$q_i({y}) \in \mathbb Z$ replace $p_i({y}) \in \mathbb Q$ for all ${y} \in R_j^\bA$. Explicitly,
\begin{align} \sum_{a \in A}r_i(a) &= 1&&\text{for all }i\in[n]\label{eq:6}\\
\sum_{{y} \in R_j^\bA} q_j({y}) &= 1&&\text{for all }j \in [m]\label{eq:7}\\
\sum_{\substack{{y} \in R_j^\bA\\{y}|_i = a}} q_j({y}) &=
r_i(a)&&\parbox{4.8cm}{for all $i\in[n], a \in A, j \in [m]$ with
$x_i$ in $\bar{x}^j$}.\label{eq:8}
\end{align} We let $\Aff_\ZZ(\bX, \bA)$ denote the integral lattice of
solutions.

\section{BLP+Affine Algorithm and Analysis for Symmetric Polymorphisms}

In the BLP+Affine algorithm, given an instance $\bX$ of $\PCSPD(\bA, \bB)$,
we seek to throw out any assignment to a constraint for which the LP determines to have weight $0$.
That is, 
given a relative interior point $(w,p)$ of $\LP_\QQ(\bX, \bA)$,
we refine $\Aff_\ZZ(\bX, \bA)$ to $\Aff_\ZZ'(\bX, \bA)$
by requiring $r_i(a)$ to be zero whenever $w_i(a)$ is,
and requiring $q_i(y)$ to be zero whenever $p_i(y)$ is
(by adding equations or just removing those variables from equations defining $\Aff_\ZZ(\bX, \bA)$).

The algorithm is presented in Figure~\ref{fig:Alg}. Note it does not depend on $\bB$; it is only relevant for the correctness proof.

\begin{figure}[h]
\begin{framed}
  \begin{enumerate}
  \item Find a relative interior point in $\LP_{\mathbb
Q}(\bX, \bA)$. If no solution exists, \textbf{Reject}.
  \item Refine $\Aff_\ZZ(\bX, \bA)$ to $\Aff_\ZZ'(\bX, \bA)$ by throwing out
assignments to constraints which have weight $0$ according to the relative interior point.
  \item If $\Aff_\ZZ'(\bX, \bA)$ is empty,
\textbf{Reject}. Else, \textbf{Accept}.
  \end{enumerate}
\end{framed}
\caption{BLP+Affine algorithm}
\label{fig:Alg}
\end{figure}

\begin{definition}\label{def:blpWorks}
	We say the BLP+Affine algorithm \emph{correctly solves} $\PCSPD(\bA,\bB)$
	if it accepts any instance $\bX$ satisfiable in $\bA$ 
	and rejects any instance unsatisfiable in $\bB$.
\end{definition}

As stated in the introduction, both the algorithm and the proof are structured similarly to those of \cite{DBLP:conf/innovations/KunOTYZ12} and \cite{BrakensiekGuruswami2019}. Like in those works, the weights of the LP solution and affine relaxation are used to construct a list of assignments which are plugged into the relevant polymorphism. The novel contribution here is that a single argument can cover any infinite symmetric family of polymorphisms.

\begin{theorem}\label{thm:sym} Let $(\bA, \bB)$ be a promise template (over any finite
domain) such that $\Pol(\bA, \bB)$ has symmetric polymorphisms of
arbitrarily large arities. Then, the BLP+Affine algorithm correctly solves $\PCSPD(\bA, \bB)$.
\end{theorem}

\begin{proof}
If an instance $\bX$ is satisfiable in $\bA$, then the Basic LP relaxation has a solution.
The refinement $\Aff_\ZZ'(\bX, \bA)$ includes every possible assignment which is in the
support of some LP solution, including integral solutions.
Thus it is non-empty and therefore the algorithm accepts.
  
Conversely, suppose the algorithm accepts, meaning both $\LP_{\mathbb
Q}(\bX, \bA)$ and $\Aff_\ZZ'(\bX, \bA)$ have solutions
$(w,p)$ over $\QQ_{\geq 0}$ and $(r,q)$ over $\ZZ$.
The latter is a solution of $\Aff_\ZZ(\bX, \bA)$ such that
\begin{align*}
	w_i(a) = 0 &\implies r_i(a) = 0 \quad\text{ for }i \in [n], a \in A\text{ and } \\
	p_j(y) = 0 &\implies q_j(y) = 0 \quad\text{ for }j \in [m], y \in R_j^\bA.
\end{align*}
We claim $\bX$ is satisfiable in $\bB$.
Among all the coordinates in the LP solution--the $w$'s and $p$'s--let $\ell$ be the
least common denominator of these rational numbers. Let $M$ be the
maximum absolute value of any integer which appears in the affine
solution (both the variable weights $r$ and the constraint weights $q$). Let
$f : A^L \to B$ be a symmetric polymorphism of arity
$L {}\geq M\ell^2.$ Now write $L = u\ell + v$ where
$u \in \mathbb Z_{\geq 0}$ and $v \in \{0, \hdots, \ell - 1\}$. Note that
$u \ge M\ell$.
  
For each
$i \in [n]$ and $a \in A$, let
  \[ W_i(a) := u\ell w_i(a) + vr_i(a).
  \]

  This is an integer by choice of $\ell$. For a fixed $i \in [n]$, note that by Eq.~\eqref{eq:3} and \eqref{eq:6}
  \[ \sum_{a \in A} W_i(a) = \sum_{a \in A}u\ell w_i(a) + vr_i(a) = u\ell + v = L.
  \] Also, for fixed $i \in [n]$ and $a \in A$, either $w_i(a) = 0$,
  which implies that $r_i(a) = 0$ by the refinement, so $W_i(a) =
  0$. Otherwise, $w_i(a) \ge 1/\ell$, so
  \[ W_i(a) \ge u\ell(1/\ell) + v(-M) \ge M\ell
- \ell M {}={} 0.
  \]
 That is, $W_i(a)$ for $a \in A$ are non-negative integers which sum to $L$.
  We claim that the assignment
  \[ X_i := f(\hdots, \underbrace{a, \hdots, a}_{W_i(a)\text{ times } \forall a \in A}, \hdots)
  \] to $x_i$ defines a satisfying assignment of $\bX$ in $\bB$.
  (Since $f$ is symmetric, only the quantity of each $a \in A$ in the input matters.)
  To verify it is indeed satisfying, consider a constraint in $(R_j,\bar{x}^j)$ (with $j \in [m]$)
  and assume without loss of generality it is on variables $\bar{x}^j=(x_1, \hdots, x_k)$.
  We claim $(X_1,\dots,X_k) \in R_j^\bB$.

  For every valid assignment $y \in R^{\bA}_j$ to that constraint in $\bA$, define
  \[
    P_j(y) := u\ell p_j(y) + v q_j(y).
  \]
  By similar logic as before, these are non-negative integers that sum to 1.
  Indeed, by Eqs.~\eqref{eq:4} and~\eqref{eq:7},
  \[\sum_{y \in R_j^\bA}P_j(y) = u\ell\sum_{y \in R_j^\bA}p_j(y) + v\sum_{y
    \in R_j^\bA}q_j(y) = L.\]
  Moreover, either $p_j(y) = q_j(y) = 0$, implying $P_j(y) = 0$, or
  \[
    P_j(y) \ge u\ell(1/\ell) + v(-M) \ge M\ell - \ell M {}={} 0.
  \]
  Further note that by Eqs.~\eqref{eq:5} and~\eqref{eq:8},
  \begin{align}
    W_i(a) & {} = u \ell w_i(a) + v r_i(a) \nonumber \\
    &= u \ell \sum_{\substack{y \in R_j^\bA\\y|_i=a}} p_i(y) + v
        \sum_{\substack{y \in R_j^\bA\\y|_i=a}} q_i(y) \nonumber \\
    &= \sum_{\substack{y \in R_j^\bA\\y|_i = a}} P_j(y)\label{eq:9}
  \end{align}

  For each $j \in [m]$ consider a matrix $M(j) \in A^{L \times k}$,
  where exactly $P_j(y)$ of the rows are equal to $y$.
  For all $i \in [k]$ and
  $a \in A$, the number of times that $a$ appears in column $i$
  is precisely $W_i(a)$ by Eq.~\eqref{eq:9}.  Thus, $f$ applied to the
  columns is precisely $(X_1, \hdots, X_k)$.
  Since $f$ is a polymorphism, this implies $(X_1, \hdots, X_k) \in R_j^\bB$.
  This concludes the proof that assigning the value $X_i$ to each variable $x_i$ (for $i \in [n]$)
  satisfies $\bX$ in $\bB$
  and hence that the algorithm is correct.
\end{proof}

\emph{Remark.}
  Another algorithm which works is to solve $\LP_{\ZZ[\sqrt{2}]}(\bX, \bA)$ (that is the constrained variables are over non-negative elements of the ring $\ZZ[\sqrt{2}]$) using the algorithm from~\cite{BrakensiekGuruswami2019}, instead of $\LP_\QQ(\bX, \bA)$. In this case, Steps 2 and 3 can be omitted. To sketch why this works, it suffices to justify why solving $\LP_{\mathbb Z[\sqrt{2}]}(\bX, \bA)$ also solves $\LP_\QQ(\bX, \bA)$ and $\Aff_\ZZ'(\bX, \bA)$. For each assigned value of the form $a + b\sqrt{2}$ in a relative interior solution to $\LP_{\mathbb Z[\sqrt{2}]}(\bX, \bA)$, consider changing this variable to $a + b\eta$, where $\eta$ is a sufficiently good rational approximation of $\sqrt{2}$. Such an assignment is in the relative interior of $\LP_\QQ(\bX, \bA)$ as any inequality non-trivially involving $\eta$, in particular \eqref{eq:1}, \eqref{eq:2}, is not tight due to $\sqrt{2}$ being irrational. To see why $\Aff_\ZZ(\bX, \bA)$ is also satisfied, replace each assigned value of $a + b\sqrt{2}$ with $a$. By inspection, this assignment (when changing $w_i$'s to $r_i$'s and $p_j$'s to $q_j$'s) satisfies $\Aff_\ZZ(\bX, \bA)$. It also satisfies $\Aff_\ZZ'(\bX, \bA)$ because $a + b\sqrt{2} = 0$ with $a$ and $b$ integral implies $a = 0$.

\section{Extension of Analysis to Block Symmetric Polymorphisms}\label{sec:block}

We say that a map $f : A^L \to B$. is \emph{block-symmetric} if
there exists a partition of the coordinates of $f$ into blocks $B_1
\cup \cdots \cup B_k = [L]$ such that $f$ is permutation-invariant
within each coordinate block $B_i$. We define the \emph{width} of $f$ to be
the minimum size of any block.\footnote{Note that a function $f$ might
  have different partitions into symmetric blocks; we define the width to be the
maximum width over all such partitions. In particular, every $f : A^L \to B$ is block-symmetric with width at least~1.
Finding the exact width or an appropriate partition into blocks is non-trivial. However, we avoid computing or evaluating $f$ altogether by only considering decision problems; see Section~\ref{sec:ques} for a discussion of search problems.\looseness=-1} A natural example of a block symmetric
polymorphism with nontrivial width is \emph{alternating threshold} first studied in \cite{DBLP:conf/soda/BrakensiekG18}
\[
  AT(x_1, \hdots, x_L) = 1[x_1 - x_2 + x_3 - \cdots \pm x_L \ge 1].
\]
In this case, the blocks are the odd and even coordinates. This
polymorphism arises in the context of $\bA$ corresponding to 1-in-3
SAT and $\bB$ corresponding to NAE-SAT. Recent work shows that this
PCSP, although tractable and simple to state, is not algebraically reducible (via so-called pp-constructions) to any
tractable finite-domain CSP~\cite{BartoBulinKrokhinEtAl2019}.

We now show an analogue of Theorem~\ref{thm:sym} for block-symmetric
polymorphisms. Remarkably, the algorithm is identical to the one for ordinary
symmetric polymorphisms and is independent of the number of blocks. In particular, it could be that the Promise CSP has finitely many polymorphisms for any particular number of blocks, yet has infinitely many block-symmetric polymorphisms of increasing width.

As discussed in \cite{BrakensiekGuruswami2019,FicakKozikOlsakEtAl2019}, nearly all known tractable Boolean PCSPs have polymorphisms which are either symmetric (such as threshold functions) or block-symmetric (such as alternating threshold). Thus, except for those PCSPs which are ``homomorphic relaxations''\footnote{A homomorphic relaxation of a $\PCSP(\bA,\bB)$ is another $\PCSP(\bC,\bD)$ such that $\bC$~has~a~homomorphism~to~$\bA$ and $\bB$ to $\bD$. In this case, $\PCSP(\bC,\bD)$ trivially reduces to $\PCSP(\bA,\bB)$. In general, if $(\bC,\bD)$ is a Boolean template that is a homomorphic relaxation of a tractable non-Boolean (P)CSP template, then this is the only algorithm we know for $\PCSP(\bC,\bD)$. We leave as an open question finding an explicit Boolean PCSP which is a homomorphic relaxation of a non-Boolean CSP but not correctly solvable by our BLP+Affine algorithm.} of a  tractable (P)CSP (c.f., \cite{BrakensiekGuruswami2019,BartoBulinKrokhinEtAl2019}), the algorithm presented here supersedes those works in the context of decision PCSP.

\begin{theorem}\label{thm:block-sym} Let $(\bA, \bB)$ be a promise template (over any finite
domain) such that $\Pol(\bA, \bB)$ has block-symmetric polymorphisms of
arbitrarily large width. Then, the BLP+Affine algorithm correctly solves $\PCSPD(\bA, \bB)$.
\end{theorem}

\begin{proof} The proof proceeds much like that of Theorem~\ref{thm:sym}. As before, we know that if 
$\bX$ is satisfiable in $\bA$, then the algorithm rejects.
We seek to show that if the algorithm accepts, then $\bX$ is satisfiable in $\bB$.

  Again, let $\ell$ be the least common denominator of all coordinates in
the LP solution. Let $M$ be the maximum absolute value of any integer which
appears in the affine solution. Let $f : A^{B_1 \cup \cdots \cup B_\kappa} \to B$ be a block-symmetric
polymorphism such that each block $B_b$, with $b \in [\kappa]$, has size at least
$M\ell^2$. Let $L_b = |B_b|$.  Similar to before, for all $b \in [\kappa]$, write $L_b = u_b\ell + v_b$ where $u_b \in \mathbb Z_{\geq 0}$ and $v \in \{0, \hdots, \ell - 1\}$. Note that
$u_b \ge M\ell$.

  We seek to show there exists \emph{a homomorphism from $\bX$ to $\bB$}. For each
$b \in [\kappa]$, $i \in [n]$ and $a \in A$, let
  \[ W_{b,i}(a) := u_b\ell w_i(a) + v_br_i(a).
  \]

  For a fixed $b \in [\kappa]$ and $i \in [n]$, by similar logic to the proof of Theorem~\ref{thm:sym}, we have that $W_{b,i}(a)$
  are non-negative integers for all $a \in A$ and
  \[\sum_{a \in A} W_{b,i}(a) = \sum_{a \in A}\big(u_b\ell w_i(a) + v_br_i(a)\big) = u_b\ell + v_b = L_b.\]

  We now claim that the assignment
  \[ X_i := f(\underbrace{\hdots, \underbrace{a, \hdots, a}_{W_{1,i}(a)\text{ times}}, \hdots}_{L_1\text{ total}}, \hdots, \underbrace{\hdots, \underbrace{a, \hdots, a}_{W_{\kappa,i}(a)\text{ times}}, \hdots}_{L_\kappa\text{ total}})
  \] to $x_i$ defines a satifying assigment of $\bX$ in $\bB$.
  To verify this, 
  consider a constraint in $(R_j,\bar{x}^j)$ (with $j \in [m]$)
  and assume without loss of generality it is on variables $\bar{x}^j=(x_1, \hdots, x_k)$.
  We claim $(X_1,\dots,X_k) \in R_j^\bB$.
  For all $b \in [\kappa]$ and  assignments $y \in R_j^\bA$ define
  \[
    P_{b,j}(y) := u_b\ell p_j(y) + v_b q_j(y).
  \]
  By similar logic as previously, $P_{b,j}(y)$ are non-negative integers and 
  by Eqs.~\eqref{eq:4} and~\eqref{eq:7},
  \[\sum_{y \in R_j^\bA}P_{b,j}(y) = u_b\ell\sum_{y \in R_j^\bA}p_j(y) + v_b\sum_{y
    \in R_j^\bA}q_j(y) = L_b.\]    
  Further note that by Eqs.~\eqref{eq:5} and~\eqref{eq:8} for $i\in[n], a \in A,$ and $j\in[m]$
  \begin{align}
    W_{b,i}(a) &= u_b \ell \sum_{\substack{y \in R_j^\bA\\y|_i=a}} p_j(y) + v_b
    \sum_{\substack{y \in R_j^\bA\\y|_i=a}} q_j(y) \nonumber \\
    &= \sum_{\substack{y \in R_j^\bA\\y|_i = a}} P_{b,j}(y)\label{eq:10}
  \end{align}
  
  For each $j \in [m]$ consider a matrix $M(j) \in A^{L \times k}$,
  where exactly $P_{b,j}(y)$ of the rows are equal to $y$ in the rows
  indexed by block $B_b$.
  For all $i \in [k]$ and $a \in A$, the number of times
  that $a$ appears in column $i$ and row-block $B_b$ is
  precisely $W_{b,i}(a)$ by Eq.~\eqref{eq:10}.  Thus, $f$ applied to
  the columns is precisely $(X_1, \hdots, X_k)$.
  Since $f$ is a polymorphism, this implies $(X_1, \hdots, X_k) \in R_j^\bB$.
  This concludes the proof that the algorithm is correct.
\end{proof}

\section{Characterizing the Algorithm's Power}\label{sec:char}
In this section, we characterize the power of the BLP+Affine algorithm from Figure~\ref{fig:Alg} exactly.
Recall, we denote the domains of relational structures $\bA,\bB,\bX$ as $A,B,X$.

\begin{theorem}\label{thm:char}
	Let $(\bA,\bB)$ be a promise template. The following are equivalent:
	\begin{itemize}
		\item BLP+Affine algorithm correctly solves $\PCSPD(\bA,\bB)$.
		\item $\Pol(\bA,\bB)$ has block-symmetric polymorphisms of arbitrarily high width.
		\item For every $L\in\NN$, $\Pol(\bA,\bB)$ has a block-symmetric polymorphism of arity $2L+1$ with two symmetric blocks of variables of size $L$ and $L+1$, respectively.
	\end{itemize}
\end{theorem}

\noindent
We need a few definitions and fundamental facts from~\cite{BulinKrokhinOprsal2019,BartoBulinKrokhinEtAl2019}.
For an $L$-ary function $f \colon A^L \to B$ and a function $\pi \colon [L] \to [L']$,
the \emph{minor} of $f$ obtained from $\pi$ is the function $g \colon A^{L'} \to B$ defined as
\vspace*{-0.2\baselineskip}
\begin{equation}\label{eq:minor}
 g(x_1,\dots,x_{L'}) := f(x_{\pi(1)}, \dots, x_{\pi(L)}).
\end{equation}
We write $g = f_{/\pi}$. 
Thus sets of polymorphisms $\Pol(\bA,\bB)$ are equipped with an operation $(\cdot)_{/\pi}$
which maps $L$-ary polymorphisms to ${L'}$-ary polymorphisms (for every $\pi \colon [L] \to [L']$).
We consider such a structure more abstractly, allowing any objects to play the role of polymorphisms:

\begin{definition}
	A \emph{minion} $\Mm$ consists of sets $\Mm^{(L)}$ for $L \in \NN$
	and functions $(\cdot)_{/\pi} \colon \Mm^{(L)} \to \Mm^{(L')}$
	for all functions $\pi \colon [L] \to [L']$,
	such that compositions agree:
	$(f_{/\pi})_{/\tau} = f_{/\tau \circ \pi}$ for $\pi \colon [L] \to [L']$, $\tau \colon [L'] \to [L'']$,
	and $f_{/\id}=f$.
	We write $\Mm$ for the disjoint union of $\Mm^{(L)}$, $L \in \NN$,
	and $\ar(f) = L$ for $f \in \Mm^{(L)}$.
	A \emph{minion homomorphism} $\xi \colon \Mm \to \Nn$ is a function which preserves arity and minors:
	$\ar(\xi(f)) = \ar(f)$ and $\xi(f_{/\pi}) = \xi(f)_{/\pi}$ for all functions $\pi \colon [L] \to [L']$.
\end{definition}
Note that the objects in a minion do not have to be functions, and the set $\Mm^{(L)}$ does not have to be finite, though this is true for minions $\Pol(\bA,\bB)$ with finite $\bA,\bB$.
Similarly the operations $(\cdot)_{/\pi}$ are not necessarily defined by Eq.~\eqref{eq:minor}, though this will always be the case when elements of a minion $f \in \Mm^{(L)}$ are $L$-ary function.
As an important example, consider the minion $\Qconv$ of convex combination functions, i.e. functions $\QQ^L \to \QQ$ of the form $w_1 x_1 + \dots + w_L x_L$ for $\sum_1^L w_i = 1$, $w_i \in \QQ_{\geq 0}$, with $(\cdot)_{/\pi}$ defined by Eq.~\eqref{eq:minor}.
We can describe the same minion more concisely by identifying a convex $L$-ary function with its $L$-tuple of coefficients $(w_1,\dots,w_L)$.
That is, the ``$L$-ary objects'' of the minion $\Qconv$ can be equivalently defined as distributions on $[L]$:
\[\Qconv^{(L)} = \{ w \colon [L] \to \QQ_{\geq 0} \  \mid \  \textstyle\sum_{i \in [L]} w(i) = 1\},\]
and for $\pi \colon [L] \to [L']$ and $w \in \Qconv^{(L)}$ one can define $w_{/\pi}$ as
\[ w_{/\pi}(i) := w(\pi^{-1}(i)) = \textstyle\sum_{j \in \pi^{-1}(i)} w(j) \quad \text{for }i \in [L'].\]%
\noindent%
This minion characterizes the power of the basic linear programming relaxation 
in the sense that BLP correctly solves $\PCSPD(\bA,\bB)$ (i.e. feasibility of $\LP_\QQ(\bX,\bA)$ implies $\bX$ is satisfiable in $\bB$ for all instances $\bX$) if and only if $\Qconv$ admits a minion homomorphism to $\Pol(\bA, \bB)$.
This was shown by Barto et al.~\cite[Theorem~7.9]{BartoBulinKrokhinEtAl2019}.
Our proof straightforwardly extends this part of the argument.

We first define the minion that plays the role of $\Qconv$ for the BLP+Affine relaxation.
It assigns two coefficients to every coordinate $i \in [L]$.

\begin{definition}\label{def:MBA}
	The minion $\MBA$ is defined as follows: for $L \in \NN$,
	its ``$L$-ary objects'' are
	\begin{alignat*}{4}
	\MBA^{(L)} := \{(w,r) \mid \ 
	& w \colon [L] \to \QQ_{\geq 0}, \quad & \textstyle\sum_{i \in [L]} w(i) = 1 &\\
	& r \colon [L] \to \ZZ, & \textstyle\sum_{i \in [L]} r(i) = 1 &\\
	&& \forall_{i\in [L]} \quad w(i)=0 \implies r(i)=0 & \ \}.	
	\end{alignat*}
	Equivalently, these could be seen as a function from $[L]$ to $\{(a, b) \in \QQ_{\geq 0} \times \ZZ : a = 0 \implies b = 0\}.$
	
	\noindent
	For $\pi \colon [L] \to [L']$ and $(w,r) \in \MBA^{(L)}$,
	we define the minor $(w,r)_{/\pi}$ as $(w', r')$, where
	\begin{alignat*}{3}
	w'(i) &:= w(\pi^{-1}(i)) = \textstyle\sum_{j \in \pi^{-1}(i)} w(j) & \\
	r'(i) &:= r(\pi^{-1}(i)),  & \text{ for } i\in[L'].
	\end{alignat*}
\end{definition}

It is easy to check this indeed defines a minion (the $ w(i)=0 \implies r(i)=0$ condition is preserved when taking a minor and composition of minors works as expected).
One could also think of a pair $(w,r) \in \MBA^{(L)}$ as an $L$-ary function on $\QQ^2$, $f(\binom{x_1}{y_1},\dots,\binom{x_n}{y_n}) = \binom{\sum w(i) x_i}{\sum r(i) x_i}$.

The minion $\MBA$ characterizes the BLP+Affine relaxation as follows.

\begin{lemma}\label{lem:solvesIffMinionHom}
	Let $(\bA, \bB)$ be a promise template.
	The following are equivalent:
	\begin{itemize}
		\item BLP+Affine correctly solves $\PCSPD(\bA,\bB)$ (Definition~\ref{def:blpWorks}).
		\item $\MBA$ admits a minion homomorphism to $\Pol(\bA,\bB)$.
	\end{itemize}
\end{lemma}

As the proof of this lemma directly extends the arguments by Barto et al.~\cite{BartoBulinKrokhinEtAl2019},
we refer the reader to Appendix~\ref{app:qconv} for an exposition of it.

We now reinterpret this last condition in terms of concrete polymorphisms.
One direction is simple:

\begin{lemma}\label{lem:minionHomImpliesPolym}
	Suppose $\MBA$ has a minion homomorphism to some minion $\Nn = \Pol(\bA, \bB)$.
	Then for every $L \in \NN$, $\Nn$ contains a block-symmetric polymorphism of arity $2L+1$ with two blocks of size $L$ and $L+1$.
\end{lemma}
\begin{proof}
	Given $L \in \NN$, consider the following object $(w,r) \in \MBA^{(2L+1)}$:
	take $w(i) := \frac{1}{2L+1}$ and $r(i) := (-1)^{i+1}$ for $i=1,\dots,2L+1$.
	For every permutation $\pi \colon [2L+1] \to [2L+1]$ which maps odd coordinates to odd coordinates (and even to even), $(w,r)_{/\pi} = (w,r)$.
	Thus the image of $(w,r)$ in $\Nn$  has the same property, i.e. it has arity $2L+1$ and it is symmetric on odd coordinates as well as on even coordinates.
\end{proof}
We remark the above lemma in fact applies to any minion $\Nn$, not only those of the form $\Pol(\bA, \bB)$; one can define $f \in \Nn^{(L)}$ to be block-symmetric with blocks $B_1 \cup \cdots \cup B_k = [L]$ if $f_{/\pi}=f$ holds for all permutations $\pi$ of $[L]$ that preserve the blocks; the proof then applies without change.

The idea for the other direction is essentially the same as in the proof of Theorem~\ref{thm:sym} and~\ref{thm:block-sym}.
We apply it to construct a minion homomorphism from every finite subset of $\MBA$ and use a compactness argument.

\begin{lemma}\label{lem:polymImpliesMinionHom}
	Suppose the minion $\Nn=\Pol(\bA,\bB)$ (for $\bA,\bB$ finite) contains block-symmetric polymorphisms of arbitrarily high width.
	Then $\MBA$ admits a minion homomorphism to $\Nn$.\looseness=-1
\end{lemma}
\begin{proof}
	To avoid cumbersome notation we present the proof only for the case of one block,
	i.e. we assume that $\Nn$ contains symmetric polymorphisms of arbitrarily high arity.
	This extends to more blocks just as Theorem~\ref{thm:block-sym} extends Theorem~\ref{thm:sym}.
	
	We define finite subsets of $\MBA$ as follows.
	For $L,\ell,M \in \NN$, let $\Mm^{(L)}_{\ell,M}$ be the subset of those $(w,r) \in  \MBA^{(L)}$ such that $\ell w(i) \in \ZZ$  for $i \in [L]$ and $\sum_i |r(i)| \leq M$.
	Observe that $\Mm^{(L)}_{\ell,M} $ is a finite set (since the numbers $\ell w(i)$ are $L$ non-negative integers summing to $\ell$ and the numbers $r(i)$ are $L$ integers between $-M$ and $M$).	
	Denote $\Mm_{\ell,M} := \bigcup_{L \in \NN}  \Mm^{(L)}_{\ell,M}$.	
	
	For fixed $\ell,M$, we define a minion homomorphism from $\Mm_{\ell,M}$ to $\Nn$ as follows.
	Let $f \in \Nn$ be a function of some arity $L^* \geq M \ell^2$.
	Let $u,v \in \NN$ be numbers such that $L^* = u \ell + v$, $v \in \{0,\dots,\ell-1\}$.
	Then $u \geq M \ell$.
	
	Take $L \in \NN$ and $(w,r) \in \Mm^{(L)}_{\ell,M}$.
	For $i \in [L]$, the number $W_i := u \ell w(i) + v r(i)$ is a non-negative integer.
	Since $\sum_i W_i = u \ell + v = L^*$,
	we can map $(w,r)$ to the $L$-ary minor $g := f(x_1, x_1, x_1, \dots, x_{L}, x_{L})$ of the $L^*$-ary function $f$ where $x_i$ is repeated $W_i$ times, for $i \in [L]$.
	We claim that this map is a minion homomorphism from $\Mm_{\ell,M}$ to $\Nn$
	(in fact to the subminion of minors of $f$).
	Indeed, for $\pi \colon [L] \to [L']$,
	consider the minor $g_{/\pi}$ of $g$ identifying $x_j$ for $j \in \pi^{-1}(i)$ into a single variable $z_i$ (for $i \in [L']$).
	We have that $g_{/\pi}$ is also a minor of $f$ where $z_i$ is repeated $\sum_{j \in \pi^{-1}(i)} W_j$ times.
	That is, $z_i$ is repeated $u \ell w(\pi^{-1}(i)) + v r(\pi^{-1}(i))$ times.
	By symmetry of $f$ the ordering does not matter, thus $g_{/\pi}$ (the minor of the image of $f$) is the same as the image of the minor $f_{/\pi}$.
	
	We conclude with a compactness argument similar to that of Remark 7.13 in \cite{BartoBulinKrokhinEtAl2019}. For $k \in \NN$, let $\Mm_k := \bigcup_{L\leq k} \Mm^{(L)}_{k!,k}$.
	Then $\Mm_k$ is finite, $\Mm_{k} \subseteq \Mm_{k+1}$ (because $k! \cdot w(i) \in \ZZ$ implies $(k+1)! \cdot w(i) \in \ZZ$) and $\bigcup_{k \in \NN} \Mm_k = \MBA$.
	Consider the possible minion homomorphisms from $\Mm_{k}$ to $\Nn$, or more precisely,
	restrictions of homomorphisms obtained above to $\Mm_{k}$
	(since $\Mm_k$ itself is technically not a minion).
	There are only finitely many possible such restrictions $\Mm_{k} \to \Nn$,
	because $\Mm_{k}$ is finite, the arities of images in $\Nn$ are bounded,
	and hence the number of possible images in $\Nn$ is also finite.
	Consider an infinite tree with restrictions from any $\Mm_k$ to $\Nn$ as nodes,
	the trivial map from $\Mm_0 = \emptyset$ being the root,
	and the parent of a function $\Mm_{k+1} \to \Nn$ being its restriction to $\Mm_{k}$.
	This is an infinite tree (because for each $k$ we have some minion homomorphism from a superset of $\Mm_k$ to $\Nn$) that is connected (because everyone is connected through its ancestors to the root)
	and finitely branching (because there are only finitely many restrictions $\Mm_k \to \Nn$, for any fixed $k$).
	Therefore, by K\H{o}nig's lemma, the tree contains an infinite path $\zeta_k \colon \Mm_k \to \Nn$ of homomorphisms that are restrictions of each other.
	Their union is then a homomorphism from $\bigcup_{k \in \NN} \Mm_k = \MBA$ to $\Nn$.
\end{proof}

(We remark the above proof in fact applies to any minion $\Nn$, assuming $\Nn^{(L)}$ is finite for every $L$.)
Lemmas~\ref{lem:solvesIffMinionHom}, \ref{lem:minionHomImpliesPolym}, and \ref{lem:polymImpliesMinionHom} conclude the proof of Theorem~\ref{thm:char}.

\section{Concluding Thoughts}\label{sec:ques}

We conclude with a few natural directions of future inquiry raised by this work.

Inspecting the proofs of Theorems~\ref{thm:sym} and~\ref{thm:block-sym}, in order to yield a search algorithm (and not just a decision algorithm), it would suffice to compute:
\[ X_i := f(\hdots, \underbrace{a, \hdots, a}_{W_i(a)\text{ times}}, \hdots)\]
for some block-symmetric polymorphism $f$ and a fixed partition into blocks of size at least $L$, for an integer $L$ which
depends polynomially on the least common denominator of rational numbers in the LP solution and the maximum absolute value of integers in the affine solution.
In previous work~\cite{BrakensiekGuruswami2019}, Brakensiek and Guruswami circumvented this problem by assuming that $f$ has special structure (such as being a threshold function, etc.). Even then, we often only assumed that you had oracle access to the structure of $f$. Thus, except for some simple cases studied in the paper, truly polynomial-time search algorithms remain elusive. Perhaps one could hope for a search algorithm like the decision algorithm presented in this paper which is oblivious to the underlying polymorphisms (as long as they are symmetric/block-symmetric).

\textbf{Question.}
  Is there an ``oblivious'' polynomial-time algorithm for the search version of Promise CSPs with infinitely many symmetric polymorphisms?

We note that an oblivious polynomial-time algorithm is also not known for the
search version of Promise CSPs with symmetric polymorphisms of all arities
(which capture the power of BLP~\cite[Theorem~7.9]{BartoBulinKrokhinEtAl2019}) and for the
search version of Promise CSPs with alternating polymorphisms of all odd arities
(which capture the power of the affine
relaxation~\cite[Theorem~7.19]{BartoBulinKrokhinEtAl2019}).

Otherwise, one could hope to prove a ``structure theorem'' that every Promise CSP with infinitely many symmetric polymorphisms also has an infinite threshold-periodic family. As \cite{BrakensiekGuruswami2019} shows, such polymorphisms can get exceedingly complicated, suggesting that such a characterization may only be possible in the Boolean case.

\textbf{Question.}
  Does every Boolean PCSP with infinitely many symmetric polymorphisms have an infinite threshold-periodic family? 

Even without a structure theorem, one could perhaps hope to compute the pertinent values of $f$ ``on the fly,'' but this seems difficult in our current formulation as the arity of $f$ could be exponentially large in the input size!

While Theorem~\ref{thm:char} characterizes the power of the BLP+Affine algorithm, it is still worthwhile to ask how this compares to other classes of templates, in particular those studied for non-promise CSPs.
The following example of a simple template not solved by the BLP+Affine relaxation was communicated to us by Jakub Opr\v{s}al.
\begin{example}
	Let $\bA$ be the disjoint union of a directed 2-cycle $\{0,1\}$ and a directed 3-cycle $\{0',1',2'\}$.
	Then $\bA$ is tractable template (i.e. $\PCSP(\bA,\bA)$ is solvable in polynomial time, in fact $\Pol(\bA,\bA)$ has cyclic polymorphisms of every prime arity $p>3$) but has no non-trivial block-symmetric polymorphisms.
\end{example}
\begin{proof}
	To see it admits no block-symmetric polymorphisms $f$ of width greater than one,
	observe that every such width can be represented as $2n+3n'$ for some $n,n' \in \NN$,
	hence every block can be filled with $n$ copies of values $0,1$ and $n'$ copies of $0',1',2'$, giving some input $\bar{v}$ to $f$.
	But $f$ should give the same output on the input $\bar{v}^{\oplus 1}$ consisting of $n$ copies of $1,0$ and $n'$ copies of $1',2',0'$.
	Since $(v_i, v^{\oplus 1}_i)$ is an arc of $\bA$ for every $i$ and since $f$ is a polymorphism, $(f(\bar{v}),f(\bar{v}^{\oplus 1}))$ would be a loop in $\bA$, a contradiction.

We now observe that $\PCSP(\bA,\bA)$ has a straightforward polynomial time algorithm. For each connected component of constraints, the variables must map to either $\{0, 1\}$ or $\{0',1',2'\}$. The first case is equivalent to testing if the graph of constraints is bipartite. The latter can be done by a breath-first search which checks that all directed cycles have length a multiple~of~$3$.
\end{proof}

Thus the condition of having block-symmetric polymorphisms of high width is not preserved under disjoint union, even though tractability is.
We also know that since $\Pol(\bA,\bA)$ has a majority polymorphism (simply let
$f(x,y,z)$ output $x$ if $x=y$ and $z$ otherwise), $\PCSP(\bA,\bA)$ can be
solved in polynomial time via the $(2,3)$-consistency algorithm, 3-rounds of
Sherali-Adams, or the canonical SDP relaxation (see also
\cite{BartoKozik:2014,ThapperZivny2017,BartoKrokhinWillard2017}). Informally,
these relaxations ensure that there are locally consistent assignments to every
(constant-sized) subset of variables. This consistency is quite powerful. For
instance, 2-SAT can be solved by the BLP+Affine relaxation  or 3 rounds of
Sherali-Adams, but not the BLP by itself. This suggests the tantalising
possibility that an analogous hierarchy could provide a uniform algorithm for
all tractable non-promise CSPs.

\textbf{Question.}
  Which (decision) promise CSPs can be solved via constantly many rounds of the Sherali-Adams hierarchy for the BLP+Affine relaxation? Does this capture all tractable non-promise CSPs?
  
\section*{Acknowledgments}
We thank Libor Barto, Andrei Krokhin, and Jakub Opr\v{s}al for useful comments and encouragement. We also thank anonymous reviewers for many helpful comments.

\appendix

\section{From Relaxations to Minion Homomorphisms}\label{app:qconv}
In this appendix, we recall the definition of the minion $\Qconv$ and prove Lemma~\ref{lem:solvesIffMinionHom} from Section~\ref{sec:char}.
We do this by explaining how free structures relate BLP and Affine relaxations to minions.
We carry over the notation from Section~\ref{sec:char}.

\begin{definition}
	The minion $\Qconv$ is defined as follows:
	for $L\in\NN$, the ``$L$-ary object'' of the minion are
	\[\Qconv^{(L)} := \{ \ w \colon [L] \to \QQ_{\geq 0} \  \mid \ \textstyle\sum_{i\in[L]} w(i) = 1 \ \};\]
	for $\pi \colon [L] \to [L']$ and $w \in \Qconv^{(L)}$, we define the minor $w_{/\pi}$ of $w$ as
	\[w_{/\pi}(i) = w(\pi^{-1}(i)) = \textstyle\sum_{j \in \pi^{-1}(i)} w(j). \]
\end{definition}

Let us describe how $\Qconv$ characterizes the power of the basic linear programming relaxation;
the case of BLP+Affine will be entirely analogous.
Recall that for an instance $\bX$ of $\PCSP(\bA,\bB)$, a solution to the BLP relaxation
assigns to each variable $i \in X$ a distribution $w_i \colon A \to \QQ_{\geq 0}$ with $\sum_{a \in A} w_i(a) = 1$.
It also assigns to each constraint $j$ of $\bX$ a distribution over satisfying assignments $p_j \colon R^A \to \QQ_{\geq 0}$ with sum 1.
Finally, the relaxation requires that for a variable $i$ in a constraint $j$ of $\bX$, the assignment of $a \in A$ to $i$ has value $w_i(a) = \sum_{y} p_j(y)$, where the sum runs over all satisfying assignments $y \in R^A$ of the constraint where the variable $i$ takes value $a$.

In other words, $w_i(a) = p_j(\pi^{-1}(a))$, where $\pi = \pi_{j\to i} \colon R^A \to A$ maps a satisfying assignment $y$ to the value of variable $i$ in constraint $j$.
That is, $w_i$, as an object of $\Qconv^{|A|}$, is required to be the minor of $p_j \in \Qconv^{|R^A|}$ obtained from $\pi$.
Thus the BLP relaxation of $\bX$ is satisfiable if and only if one can assign some $w_i \in \Qconv^{|A|}$
to each variable $i \in X$ so that the following holds for every constraint $j$ of $\bX$:
there is a $p_j \in \Qconv^{|R^A|}$ such that for all variables $i$ in $j$, $w_i = {p_j}{/\pi_{j\to i}}$.
This can be phrased as the existence of a homomorphism from $\bX$ to the free structure $\FF_{\Qconv}$, defined as follows.

\begin{definition}\label{def:Qfree}
	For a relational structure $\bA$ and a minion $\Mm$, the \emph{free structure} $\FF_{\Mm}(\bA)$ is
	a template with domain $\Mm^{|A|}$ (potentially infinite) and with the same signature as $\bA$.
	For each relation $R^A$ of arity $k$ in $\bA$,
	there is a relation $R^{\FF}$ of the same arity in $\FF_{\Mm}(\bA)$ defined as follows:
	$w_1,\dots,w_k \in \Mm^{(|A|)}$ are in the relation $R^{\FF}$
	if there is some $p \in \Mm^{(|R^A|)}$
	such that for each $i \in [k]$,
	$w_i = p_{/\pi_{i}}$.
	Here $\pi_{i} \colon R^A \to A$ maps $y \in R^A \subseteq A^k$ to its $i$-th coordinate.
\end{definition}

The above discussion shows that:
\begin{observation}
	 The BLP relaxation of $(\bX,\bA)$ has a solution if and only if $\bX$ is satisfiable in $\FF_{\Qconv}(\bA)$.
\end{observation}

Just as in Definition~\ref{def:blpWorks},
we say that ``BLP \emph{correctly solves} $\PCSPD(\bA,\bB)$'' if for every instance $\bX$, feasibility of the $\LP_\QQ(\bX,\bA)$ implies satisfiability of $\bX$ in $\bB$.
(Note the other direction is always trivially true: if $\bX$ is satisfiable in $\bA$, then the relaxation $\LP_\QQ(\bX,\bA)$ has a solution).
Let us write $\bX \to \bA$ if there exists a homomorphism from $\bX$ to $\bA$ (i.e. a satisfying assignment);
we can now restate the definition.

\begin{observation}
	Let $(\bA,\bB)$ be a promise template. The following are equivalent:
	\begin{itemize}	
		\item BLP correctly solves $\PCSPD(\bA,\bB)$;
		\item for every instance $\bX$, $\bX \to \FF_{\Qconv}(\bA)$ implies $\bX \to \bB$.
	\end{itemize}
\end{observation}

\noindent
Entirely analogously, we can restate what it means for BLP+Affine to solve a PCSP (Definition~\ref{def:blpWorks}),
by using the minion $\MBA$ (Definition~\ref{def:MBA}).
\begin{observation}
	Let $(\bA,\bB)$ be a promise template. The following are equivalent:
	\begin{itemize}	
		\item BLP+Affine correctly solves $\PCSPD(\bA,\bB)$;
		\item for every instance $\bX$, $\bX \to \FF_{\MBA}(\bA)$ implies $\bX \to \bB$.
	\end{itemize}
\end{observation}

The resulting condition can be simplified by a standard compactness argument.
That~is, we use the following straightforward generalization of the de Bruijn--Erd\H{o}s Theorem
(see e.g.~\cite[Theorem 8.1.3]{Diestel} for a discussion and short proofs, \cite{RorabaughTW17} for general relational structures).

\begin{lemma}[Compactness for structures]
	Let $\bF, \bB$ be relational structures with $F$ infinite and $B$ finite.	
	If every finite induced substructure of $\bF$ admits a homomorphism to $\bB$, then so does $\bF$.\looseness=-1
\end{lemma}

\noindent
That is, for a promise template $(\bA,\bB)$ and any minion $\Mm$, the following are equivalent:
\vspace*{-1ex}
	\begin{itemize}
		\item for every instance $\bX$, $\bX \to \FF_{\Mm}(\bA)$ implies $\bX \to \bB$;
		\item $\FF_{\Mm}(\bA) \to \bB$.
	\end{itemize}

A fundamental property of free structures is that the latter condition
is equivalent to the existence of a minion homomorphism, as proved by Barto et al.~\cite[Lemma 4.4]{BartoBulinKrokhinEtAl2019}.
\begin{lemma}[\cite{BartoBulinKrokhinEtAl2019}]
	Let $(\bA,\bB)$  be a promise template and let $\Mm$ be any minion. The following are equivalent:
	\begin{itemize}
		\item $\FF_{\Mm}(\bA) \to \bB$;
		\item there exists a minion homomorphism from $\Mm$ to $\Pol(\bA,\bB)$.
	\end{itemize}
\end{lemma}

Altogether, this shows that BLP+Affine solves $\PCSP(\bA,\bB)$ if and only if $\MBA$ admits a minion homomorphism to $\Pol(\bA,\bB)$.
This concludes the proof of Lemma~\ref{lem:solvesIffMinionHom} in Section~\ref{sec:char}.

We remark that Barto et al.~\cite[Theorem 7.9]{BartoBulinKrokhinEtAl2019} used the same argument to characterize the power of BLP for PCSPs.

\begin{theorem}[\cite{BartoBulinKrokhinEtAl2019}]
	Let $(\bA,\bB)$ be a promise template.
	The following are equivalent:
	\begin{itemize}
		\item BLP solves $\PCSP(\bA,\bB)$ (as in Definition~\ref{def:blpWorks}),
		\item $\forall_\bX \ \ \bX\to \FF_{\Qconv}(\bA) \implies \bX\to \bB$,
		\item $\FF_{\Qconv}(\bA) \to \bB$,
		\item $\Qconv$ admits a minion homomorphism to $\Pol(\bA,\bB)$,
		\item $\Pol(\bA,\bB)$ contains symmetric polymorphisms of every arity.
	\end{itemize}
\end{theorem}

Our argument thus only differs in the equivalence of the last two bullets, an analogue of which is proved in Section~\ref{sec:char}.
Finally, let us note that in~\cite[Theorem 7.19]{BartoBulinKrokhinEtAl2019}, the power of the Affine relaxation alone was similarly characterized by the minion $\mathcal{Z}_{\mathrm{aff}}$, defined analogously to $\Qconv$, except with integer coefficients (not necessarily non-negative): the $L$-ary objects are $r \colon [L] \to \ZZ$ such that $\sum_{i\in [L]} r(i) = 1$.

\pagebreak[3]

\bibliographystyle{alpha}
\bibliography{symmetric}

\end{document}